\documentclass[letterpaper, 10 pt, conference]{ieeeconf}  

\IEEEoverridecommandlockouts                              
\usepackage[dvips]{graphicx}
\usepackage{amsmath}
\newtheorem{theorem}{Theorem}

\newtheorem{definition}{Definition}
\newtheorem{assumption}{Assumption}

\begin{document}

\title{\LARGE \bf
$H_2$ Performance Analysis and Synthesis for Discrete-Time Linear Systems with Dynamics Determined by an i.i.d.\ Process}


\author{Yohei Hosoe, Takashi Okamoto and Tomomichi Hagiwara
\thanks{This work was supported in part by JST PRESTO Grant Number JPMJPR2127 and JSPS KAKENHI Grant Number 20K04546.}
\thanks{Y.~Hosoe, T.~Okamoto and T.~Hagiwara are with the Department of Electrical Engineering, 
        Kyoto University, Nishikyo-ku, Kyoto 615-8510, Japan
        {\tt\small hosoe@kuee.kyoto-u.ac.jp}}%
}

\maketitle
\thispagestyle{empty}
\pagestyle{empty}

\begin{abstract}

This paper is concerned with $H_2$ control of discrete-time linear systems with dynamics determined by an independent and identically distributed (i.i.d.) process.
A definition of $H_2$ norm is first discussed for the class of systems.
Then, a linear matrix inequality (LMI) condition is derived for the associated performance analysis, which is tractable in the sense of numerical computation.
The results about analysis are also extended toward state-feedback controller synthesis.

\end{abstract}

\section{Introduction}

This paper studies
$H_2$ control of discrete-time linear stochastic systems whose dynamics is determined by an independent and identically distributed (i.i.d.) process.
This class of systems are also known as discrete-time linear systems with white parameters \cite{Koning1992whitepara} since "i.i.d." implies "white" in discrete time.
The systems with an i.i.d.\ process can be seen as a discrete-time linear case of random dynamical systems \cite{Arnold-book}; see also \cite{Hosoe-TAC22} for discrete-time linear systems with general stochastic dynamics, which properly include the present systems.
The systems with an i.i.d.\ process are also closely related with switched linear systems \cite{hespanha2004uniform, lee2006uniform}, and indeed, can be viewed as those with the switching signal given by an i.i.d. signal (whose support may be uncountable).
This is a relationship similar to that between switched systems and Markov jump systems \cite{Costa-book, lee2006uniform}, where the switching signal is given by a Markov chain.

In our earlier study \cite{Hosoe-TAC22}, 
relations of several notions of second-moment stability (i.e., mean square stability) and associated Lyapunov inequalities were discussed for discrete-time linear systems with a general stochastic process.
The Lyapunov inequalities about second-moment stability of any class of stochastic systems can be theoretically unified in the framework developed in this earlier study {\it provided that the systems are discrete-time and linear}, and hence, the associated results could have a large impact in the related fields.
For example, the Lyapunov inequality for the i.i.d.\ case in \cite{ogura2013generalized, Hosoe-TAC19} and that for the Markovian case in \cite{Costa-book, Costa-TAC14} are naturally covered in the framework as special cases.
Their extensions such as periodic versions are also not an exception.
As stated in \cite{Hosoe-TAC22}, however, the most general form of the Lyapunov inequality (i.e., that for the most general stochastic systems) involves conditional expectation operations that are difficult to numerically compute, and is not suitable for practical control problems.
The systems with an i.i.d.\ process are one of the tractable classes of systems in the sense of numerical computations, as well as Markov jump systems.
To pursue the usefulness of the proposed framework, we have been also studying practical linear matrix inequality (LMI) conditions for control of 
the systems with an i.i.d.\ process.
LMI conditions for $H_{2}$ performance analysis and synthesis derived in this paper correspond to results in such a direction.

The control theory for the i.i.d.\ case is compatible with networked control systems (NCSs) with randomly time-varying communication delays.
Although details are omitted, our preliminary experiments using the Internet suggested that 
i.i.d.\ processes are prospective as a model of actual communication delays.
Motivated by this, the i.i.d.\ results in \cite{Hosoe-TAC19} about stabilization were exploited in \cite{Hosoe-CDC22}  about NCSs with random communication delays.
One of the advantages of using our i.i.d.\ results is that we can directly handle unbounded spaces for the values of delays in controller synthesis even when the plant in the NCS is unstable; most of the earlier studies assumed the existence of an upper bound for the delays even when they are random (e.g., in \cite{zhang2005new,shi2009output}).
Our stabilization results have been already applied to the remote control of vehicles in \cite{Kameoka-ROCOND22}, in which an experiment using an actual plug-in hybrid electric vehicle and the Internet is reported.
Deriving LMI conditions not only for stabilization but also for $H_2$ control in the present paper
is expected to contribute to improving the associated control performance,
as in the cases with deterministic systems \cite{Boyd-book} and Markov jump systems \cite{morais2013h2}.

This paper is organized as follows.
Section~\ref{sc:sys-norm} describes
the systems to be dealt with in this paper and discusses a definition of $H_2$ norm.
Section~\ref{sc:anal} studies $H_2$ performance analysis as a step toward $H_2$ control.
Then, Section~\ref{sc:syn} extends the results about analysis and derives an LMI condition for controller synthesis.
Section~\ref{sc:exam} provides a simple numerical example, and Section~\ref{sc:concl} concludes the paper.

We use the following notation in this paper.
The set of real numbers,
that of positive real numbers,
that of positive integers
and that of non-negative integers
are denoted by ${\bf R}$, ${\bf R}_+$, ${\bf N}$ and ${\bf N}_0$, respectively.
The set of $n$-dimensional real
column vectors and that of $m\times n$ real matrices are denoted by
${\bf R}^n$ and ${\bf R}^{m\times n}$, respectively.
The set of $n\times n$ symmetric matrices and
that of $n\times n$ positive definite matrices 
are denoted by ${\bf S}^{n\times n}$ and
${\bf S}^{n\times n}_{+}$, respectively.
The identity matrix of size $n$ is denoted by $I_n$; the subscript will be
dropped when the size is obvious.
The maximum singular value is denoted by $\sigma_{\rm max}(\cdot)$.
The Euclidean norm
is denoted by $\|\cdot\|$.
The trace of a matrix is denoted by ${\rm tr}(\cdot)$.
The vectorization of a matrix in the row
direction is denoted by ${\rm row}(\cdot)$, i.e., ${\rm row}(\cdot):=[{\rm row}_1(\cdot),\ldots,{\rm
row}_m(\cdot)]$, where $m$ is the number of rows of the matrix and ${\rm row}_i(\cdot)$ denotes the $i$th row.
The Kronecker product is denoted by $\otimes$.
The expectation of a random variable is denoted by $E[\cdot]$; this notation
is also used for the expectation of a random matrix.
If $s$ is a random variable obeying the distribution ${\cal D}$,
then we represent it as $s \sim {\cal D}$.

\section{Discrete-Time Linear Systems with Dynamics Determined by an i.i.d.\ Process and $H_2$ Norm}
\label{sc:sys-norm}

\subsection{System Class}

Let us consider the $Z$-dimensional discrete-time stochastic process $\xi=\left(\xi_k\right)_{k\in {\bf N}_0}$ satisfying
the following assumption.

\begin{assumption}
\label{as:iid}
The random vector
$\xi_k$ is independent and identically distributed (i.i.d.) with
 respect to the discrete time $k\in {\bf N}_0$.
\end{assumption}

The time-invariant support of $\xi_k$ satisfying this assumption is denoted by ${\boldsymbol {\mathit\Xi}}\ (\subset {\bf R}^Z)$.
The process $\xi$ satisfying the above assumption is obviously stationary and ergodic \cite{Klenke-book}.

With such a process $\xi$, let us further consider the discrete-time linear system $G$ represented by
\begin{align}
&
x_{k+1} = A(\xi_k) x_k +  B(\xi_k) w_k, \notag\\
&
z_k = C(\xi_k) x_k +  D(\xi_k) w_k,\label{eq:fr-sys}
\end{align}
where
$x_k$, $w_k$ and $z_k$ are the state, the input and the output, respectively.
The initial state $x_0$ is supposed to be deterministically given.
In addition,
$A:{\boldsymbol {\mathit\Xi}} \rightarrow
{\bf R}^{n\times n}$,
$B:{\boldsymbol {\mathit\Xi}} \rightarrow
{\bf R}^{n\times p_w}$,
$C:{\boldsymbol {\mathit\Xi}} \rightarrow
{\bf R}^{q_z\times n}$ and
$D:{\boldsymbol {\mathit\Xi}} \rightarrow
{\bf R}^{q_z\times p_w}$
are matrix-valued Borel functions satisfying the following assumption.

\begin{assumption}
\label{as:bound}
The squares of entries of 
$A(\xi_0)$ are all Lebesgue integrable, i.e.,
\begin{align}
&
E[A_{ij}(\xi_0)^2]<\infty\ \ (\forall i, j = 1,\ldots,n).
\label{eq:as-bound}
\end{align}
Similarly, the squares of entries of
$B(\xi_0)$, $C(\xi_0)$ and $D(\xi_0)$ are also all Lebesgue integrable.
\end{assumption}

The condition (\ref{eq:as-bound}) is known as a minimal requirement for defining second-moment stability \cite{Hosoe-TAC22}.
No other structural constraints are imposed on the functions $A$, $B$, $C$ and $D$ as well as the distribution of $\xi_k$ in this paper.

\subsection{$H_2$ Norm}

As a stability notion for the aforementioned class of systems, we use the following second-moment exponential stability, which is also called mean square exponential stability \cite{Kozin-Auto69}.

\begin{definition}
\label{df:expo}
The system $G$ satisfying
Assumptions~\ref{as:iid} and \ref{as:bound}
is said to be exponentially stable
 in the second moment if 
 for $w\equiv 0$,
 there exist $a\in {\bf R}_+$ and $\lambda \in
 (0,1)$ such that
\begin{align}
&
\sqrt{E[||x_k||^2]} \leq a ||x_0|| \lambda^k\ \ \ (\forall k \in {\bf
N}_0, \forall x_0 \in {\bf R}^n).
\label{eq:exp-def}
\end{align}
\end{definition}

This stability notion is known to be characterized by a Lyapunov inequality as in the following theorem \cite{Hosoe-TAC19}.

\begin{theorem}
\label{th:stab}
Suppose that the system $G$ satisfies Assumptions~\ref{as:iid} and \ref{as:bound}.
The following two conditions are equivalent.
\begin{enumerate}
\item
The system $G$ is exponentially stable in the second moment.
\item
There exists $P\in {\bf S}_+^{n\times n}$ such that
\begin{align}
&
E[P - A(\xi_0)^T P A(\xi_0)]> 0. \label{eq:lyapunov-asm}
\end{align}
\end{enumerate}
\end{theorem}

For stable systems, this paper defines an $H_2$ norm in the time domain.
Let us consider
\begin{align}
& 
\Phi_{k_1}^{k_2}:=
\begin{cases}
    I & (k_2=k_1) \\
    A(\xi_{k_2-1})A(\xi_{k_2-2})\ldots A(\xi_{k_1}) & (k_2\geq k_1+1)
\end{cases}
\end{align}
for $k_1$, $k_2 \in {\bf N}_0$ satisfying $k_2 \geq k_1$.
Then, the expectation of the square sum of the impulse response from $k=0$ to $k=K\ (K\in {\bf N})$ is described by
\begin{align}
&s_K:=
E\left[{\rm tr}\left(D(\xi_0)^TD(\xi_0\right)\right]\notag \\
&
+\sum_{k=1}^K
E\left[{\rm tr}\left( B(\xi_0)^T(\Phi_1^k)^T C(\xi_k)^T C(\xi_k) \Phi_1^k B(\xi_0)  \right)\right].
\label{eq:sK}
\end{align}
The convergence of this series is ensured by the following theorem.
\begin{theorem}
\label{th:conv}
Suppose that the system $G$ satisfies
Assumptions~\ref{as:iid} and \ref{as:bound}.
If $G$ is exponentially stable in the second moment, then the corresponding $s_K$ in (\ref{eq:sK}) converges to a constant as $K$ tends to infinity.
\end{theorem}

\begin{proof}
Since the trace operation and the expectation operation are commutative, 
the sum (i.e., the second term) in the right-hand side of (\ref{eq:sK}) can be equivalently rewritten as
\begin{align}
&
{\rm tr}\left(
\sum_{k=1}^K
E\left[ B(\xi_0)^T(\Phi_1^k)^T C(\xi_k)^T C(\xi_k) \Phi_1^k B(\xi_0)  \right]\right).
\label{eq:SK}
\end{align}
Let us denote the sum in the trace operation in (\ref{eq:SK}) by $S_K^\prime\ (\geq 0)$.
Then, if $S_K^\prime$ converges to a constant matrix as $K\rightarrow \infty$, then $s_K$ also converges to a constant.
Hence, we show the convergence of $S_K^\prime$.

It follows from Assumptions~\ref{as:iid} and \ref{as:bound} that 
\begin{align}
E[x_k^TC(\xi_k)^TC(\xi_k)x_k]
=&
E\left[\|C(\xi_k)x_k\|^2\right] \notag \\
\leq&
E\left[\sigma_{\rm max}(C(\xi_k))^2 \|x_k\|^2\right] \notag \\
=&
E\left[\sigma_{\rm max}(C(\xi_k))^2\right] E\left[\|x_k\|^2\right]\notag \\
=&
E\left[\sigma_{\rm max}(C(\xi_0))^2\right] E\left[\|x_k\|^2\right] \notag \\
&(\forall k \in {\bf N}_0).
\end{align}
By Definition~\ref{df:expo}, this implies that there exist $a>0$ and $1>\lambda>0$ such that
\begin{align}
&
x_0^T E[(\Phi_0^{k})^T
C(\xi_k)^TC(\xi_k)
\Phi_0^{k}] x_0 \notag \\
\leq&
a^2 E\left[\sigma_{\rm max}(C(\xi_0))^2\right] 
 \lambda^{2k} (x_0^T x_0)\notag \\
&(\forall k \in {\bf N}_0, \forall x_0 \in {\bf R}^n).
\end{align}
This inequality can be rewritten as
\begin{align}
&
E[(\Phi_0^{k})^T
C(\xi_k)^TC(\xi_k)
\Phi_0^{k}] \notag \\
\leq &
a^2 E\left[\sigma_{\rm max}(C(\xi_0))^2\right]
 \lambda^{2k} I \ \ 
 (\forall k \in {\bf N}_0).
\end{align}
Since $\xi_k$ is i.i.d.\ with respect to $k$,
this further implies
\begin{align}
&
E[(\Phi_1^{k})^T
C(\xi_k)^TC(\xi_k)
\Phi_1^{k}] \notag \\
\leq &
a^2 E\left[\sigma_{\rm max}(C(\xi_0))^2\right]
 \lambda^{2(k-1)} I \ \ 
 (\forall k \in {\bf N}).
\end{align}
Multiplying $B(\xi_0)$ and its transpose and taking expectation for this inequality lead to
\begin{align}
&
E[B(\xi_0)^T (\Phi_1^{k})^T
C(\xi_k)^TC(\xi_k)
\Phi_1^{k} B(\xi_0)] \notag \\
\leq &
a^2 E\left[\sigma_{\rm max}(C(\xi_0))^2\right]
 \lambda^{2(k-1)} E[B(\xi_0)^T B(\xi_0)] \ \ 
 (\forall k \in {\bf N}).
\end{align}
The sum of the left-hand side of this inequality from $k=1$ to $k=K$ is nothing but $S_K^\prime$.
Hence,
\begin{align}
S_K^\prime 
\leq
a^2 E\left[\sigma_{\rm max}(C(\xi_0))^2\right]
 \left(\sum_{k=1}^{K}\lambda^{2(k-1)}\right) E[B(\xi_0)^T B(\xi_0)],
\end{align}
whose right-hand side converges to a constant matrix as $k\rightarrow \infty$.
This implies that $S_\infty^\prime$ is bounded.
Since the sequence of $S_K^\prime$
with respect to $K=1,2,\ldots$ is monotonically nondecreasing under the semiorder relation based on positive semidefiniteness (i.e., 
 $S_K^\prime \leq S_{K+1}^\prime$), $S_K^\prime$ also converges to a constant matrix as $K\rightarrow \infty$.
 This completes the proof.
\end{proof}

This theorem validates our defining the $H_2$ norm of the system $G$ as
\begin{align}
&\|G\|_2:=\sqrt{s_\infty}.
\label{eq:H2}
\end{align}
This definition is consistent with that for deterministic systems.
Indeed, if we consider the case where $\xi$ is given by a deterministic time-invariant constant process,
the $H_2$ norm defined in (\ref{eq:H2}) of the corresponding deterministic time-invariant $G$ coincides with the usual $H_2$ norm for deterministic time-invariant systems.
This paper deals with such an $H_2$ norm for stochastic systems.

\section{$H_2$ Performance Analysis}
\label{sc:anal}

\subsection{Expectation-Based Inequality Condition}

This section first proves the following theorem about $H_2$ analysis of stochastic systems, which is one of the main results in this paper.
\begin{theorem}
\label{th:main-H2anal}
Suppose that the system $G$ satisfies Assumptions~\ref{as:iid} and \ref{as:bound}.
For given $\gamma\in {\bf R}_+$,
the following two conditions are equivalent.
\begin{enumerate}
 \item 
 The system $G$ is exponentially stable in the second moment and satisfies $\|G\|_2<\gamma$.
\item
There exists $P\in {\bf S}^{n\times n}_+$ such that%
\begin{align}
&
E[P - A(\xi_0)^T P A(\xi_0) - C(\xi_0)^T C(\xi_0)]> 0, \label{eq:H2-anal-1}\\
& 
E[{\rm tr}(D(\xi_0)^T D(\xi_0) + B(\xi_0)^T P B(\xi_0))] < \gamma^2. \label{eq:H2-anal-2}
\end{align}
\end{enumerate}
\end{theorem}

\medskip

\begin{proof}
For notational simplicity, we use
$A_k:=A(\xi_k)$, $B_k:=B(\xi_k)$, $C_k:=C(\xi_k)$ and $D_k:=D(\xi_k)$ in this proof.

\medskip

2$\Rightarrow$1:
It follows from (\ref{eq:H2-anal-1}) and Theorem~\ref{th:stab} that the system $G$
is exponentially stable in the second moment.
Hence, it suffices to show $\|G\|_2<\gamma$.

Since $\xi_k$ is i.i.d.\ with respect to $k$ by Assumption~\ref{as:iid},
the inequality (\ref{eq:H2-anal-1}) implies
\begin{align}
&
E[A_k^T P A_k + C_k^T C_k] < P\ \ (\forall k\in {\bf N}_0). \label{eq:H2-anal-1-kai}
\end{align}
Take $K\in {\bf N}$ and consider the case of $k=K$ in the above inequality.
Then,
by multiplying $A_{k-1}\ (k=K)$ and its transpose on
the inequality, we have
\begin{align}
&\hspace{-3mm}
A_{K-1}^T E[A_K^T P A_K + C_K^T C_K]A_{K-1}
\leq A_{K-1}^T P A_{K-1}.
\label{eq:ope1}
\end{align}
Since $A_{K-1}$ is a random matrix, taking expectation for both sides of this inequality,
together with using the independence between $\xi_{K-1}$ and $\xi_{K}$, leads to
\begin{align}
&E[(\Phi_{K-1}^{K+1})^T P \Phi_{K-1}^{K+1} + A_{K-1}^T C_K^T C_K A_{K-1}] \notag \\
\leq &E[A_{K-1}^T P A_{K-1}].
\end{align}
Adding $E[C_{k-1}^T C_{k-1}]\ (k=K)$ to both sides of this inequality and using
 (\ref{eq:H2-anal-1-kai}) further lead to
\begin{align}
&E[(\Phi_{K-1}^{K+1})^T P \Phi_{K-1}^{K+1} + A_{K-1}^T C_K^T C_K A_{K-1}
 \notag \\
& + C_{K-1}^T C_{K-1}] \leq P.\label{eq:sousakekka}
\end{align}
By repeating the operation from (\ref{eq:ope1}) to (\ref{eq:sousakekka}) for $k=K-1, K-2, \ldots$,
we finally obtain
\begin{align}
&\hspace{-3mm}
E\left[(\Phi_{1}^{K+1})^T P \Phi_{1}^{K+1} + \sum_{k=1}^{K} (\Phi_{1}^k)^T C_{k}^T C_{k} \Phi_{1}^k\right] \leq P.
\end{align}
It follows from this inequality and $E[(\Phi_{1}^{K+1})^T P \Phi_{1}^{K+1}]\geq 0$ that
\begin{align}
&
E\left[\sum_{k=1}^{K} (\Phi_{1}^k)^T C_{k}^T C_{k} \Phi_{1}^k\right] \leq P.
\end{align}
Multiplying $B_0$ and its transpose and taking expectation for
this inequality lead to
\begin{align}
&\hspace{-7mm}	
E\left[\sum_{k=1}^{K} B_0^T (\Phi_{1}^k)^T C_{k}^T C_{k}
 \Phi_{1}^k B_0\right] \leq E[B_0^T P B_0].
\end{align}
Adding $E[D_0^T D_0]$ and taking trace for this inequality further lead to
\begin{align}
&\hspace{-4mm}
E[{\rm tr}(D_0^T D_0)]+
\sum_{k=1}^{K}E\left[{\rm tr}(B_0^T (\Phi_{1}^k)^T C_{k}^T C_{k}
 \Phi_{1}^k B_0)\right]\notag\\
&\hspace{-4mm}
 \leq E[{\rm tr}(D_0^T D_0 + B_0^T P B_0)].
\end{align}
The left-hand side of this inequality is nothing but $s_K$ defined in (\ref{eq:sK}).
Hence, by letting $K\rightarrow \infty$, the above inequality leads us to 
\begin{align}
&
\|G\|_2^2
 \leq E[{\rm tr}(D_0^T D_0 + B_0^T P B_0)],
\end{align}
where the convergence of the associated infinite series is ensured by Theorem~\ref{th:conv}.
This, together with (\ref{eq:H2-anal-2}), leads to $\|G\|_2<\gamma$.

\medskip

1$\Rightarrow$2:
We prove this assertion by four steps.

(Step 1)
Since the system $G$ is exponentially stable in the second moment, 
there exists $\Pi>0$ satisfying
\begin{align}
&
E[\Pi-A_0^T \Pi A_0]>0 \label{eq:H2-anal-1-Pi}
\end{align}
by Theorem~\ref{th:stab}.

(Step 2)
It follows from Assumption~\ref{as:bound} and Definition~\ref{df:expo} that
there exist $a\in {\bf R}_+$ and $\lambda\in (0,1)$ satisfying
\begin{align}
&
E[x_k^T C_k^T C_k x_k] \leq a^2 E[\sigma_{\max}(C_0)^2] (x_0^T x_0)
 \lambda^{2k} \notag \\
& (\forall k \in {\bf
N}_0, \forall x_0 \in {\bf R}^n)
\end{align}
for $w\equiv 0$.
Let $b=a\sqrt{E[\sigma_{\max}(C_0)^2]}$.
Then, the above inequality implies
\begin{align}
&
E[(\Phi_0^{k})^T C_k^T C_k \Phi_0^{k}] \leq b^{2} \lambda^{2k} I \
 \ (\forall k \in {\bf
N}_0).
\end{align}
Since $\xi_k$ is i.i.d.\ with respect to $k$, for 
$Q:=E[C_0^T C_0]$,
this implies
\begin{align}
&
E[(\Phi_{k_1}^{k_2})^T Q \Phi_{k_1}^{k_2}] \leq b^{2}
 \lambda^{2(k_2-k_1)} I \notag \\
& (\forall k_1, k_2 \in {\bf
N}_0 \ \ {\rm s.t.} \ \ k_2 \geq k_1 \geq 0). \label{eq:anal-proof-epqp}
\end{align}

(Step 3)
Define
\begin{align}
 &
P^{k_2}_{k_1}:=\sum_{k=k_1}^{k_2}
(\Phi_{k_1}^{k})^T Q \Phi_{k_1}^{k}
\end{align}
for $k_1, k_2 \in {\bf
N}_0$ satisfying $k_2 \geq k_1 \geq 0$.
This $P^{k_2}_{k_1}$ naturally satisfies
\begin{align}
P^{k_2}_{k_1} - A_{k_1}^T P^{k_2}_{k_1+1} A_{k_1}=Q. \label{eq:papaq}
\end{align}
On the other hand, since $Q\geq 0$, 
the sequence of
\begin{align}
 &
E[P^{k_2}_{k_1}]=\sum_{k=k_1}^{k_2}
E[(\Phi_{k_1}^{k})^T Q \Phi_{k_1}^{k}]
 \label{eq:E-papaq}
\end{align}
with respect to $k_2=k_1, k_1+1, \ldots$ for 
each fixed $k_1$ is monotonically nondecreasing under the semiorder relation based on positive semidefiniteness, i.e., 
 $E[P^{k_2}_{k_1}]\leq E[P^{k_2+1}_{k_1}]$.
 In addition, it follows from (\ref{eq:anal-proof-epqp}) that
\begin{align}
 &
 E[P^{k_2}_{k_1}]\leq b^2\left(\sum_{k=k_1}^{k_2}\lambda^{2(k-k_1)}\right)I.
 \end{align}
Since $\lambda<1$, the right-hand side of this inequality converges to a constant matrix as $k_2 \rightarrow \infty$.
Hence, $E[P^{k_2}_{k_1}]$ also converges to a constant matrix as $k_2 \rightarrow \infty$ for each fixed $k_1$.
Since $\xi_k$ is i.i.d., this constant matrix is independent of $k_1$, and we denote it by $P^\prime$.
Then, it follows from (\ref{eq:papaq}) that
\begin{align}
&
E[P^\prime - A_0^T P^\prime A_0 - C_0^T C_0]= 0. \label{eq:H2-anal-1-P0}
\end{align}

(Step 4)
We have obtained $\Pi>0$ satisfying (\ref{eq:H2-anal-1-Pi})
and $P^\prime\geq 0$ satisfying (\ref{eq:H2-anal-1-P0}).
With those $\Pi$ and $P^\prime$, we construct $P$ satisfying 
(\ref{eq:H2-anal-1}) and (\ref{eq:H2-anal-2}).

Let us consider the case where $E[\sigma_{\max}(B_0)^2]= 0$.
Take
$P=P^\prime+\Pi$.
Then, this $P$ is positive definite and satisfies (\ref{eq:H2-anal-1}).
In addition,
(\ref{eq:H2-anal-2})  becomes $E[{\rm tr}(D_0^T D_0))]<\gamma^2$, which is automatically satisfied under condition~1 (recall the definition of $H_2$ norm).

Let us next consider the case where
$E[\sigma_{\max}(B_0)^2]> 0$.
Take $P=P^\prime+\beta \Pi$ with
$\beta=\frac{\gamma^2- \|G\|_2^2}{2E[{\rm tr}(B_0^T \Pi B_0)]}$.
This $P$ is also positive definite and satisfies (\ref{eq:H2-anal-1}) as in the above case since $\beta>0$.
It follows from the definitions of $H_2$ norm and $P^\prime$ that
\begin{align}
\|G\|_2^2=E[{\rm tr}(D_0^T D_0) + {\rm tr}(B_0^T P^\prime B_0)].
\end{align}
Hence, with the above $P$, 
\begin{align}
E[{\rm tr}(D_0^T D_0 + B_0^T P B_0)] 
=&
\|G\|_2^2+\beta E[{\rm tr}(B_0^T \Pi B_0)] \notag \\
=&
\frac{\gamma^2 + \|G\|_2^2}{2}  \notag \\
<&
\gamma^2
\end{align}
holds.

Hence, in any case, there exists $P>0$ satisfying 
(\ref{eq:H2-anal-1}) and (\ref{eq:H2-anal-2}).
This completes the proof.
\end{proof}

\subsection{Numerically Tractable Condition}

Theorem~\ref{th:main-H2anal} gives an expectation-based inequality condition for $H_2$ analysis of stochastic systems.
This form of condition, however, is generally not tractable from the aspect of numerical computation since products involving the decision variable $P$ are in the expectation operation.
This issue can be resolved through the following theorem.

\begin{theorem}
\label{th:anal-trac}
For given $\gamma\in {\bf R}_+$ and $P\in {\bf S}^{n\times n}_+$,
the following two conditions are equivalent.
\begin{enumerate}
\item
The inequalities (\ref{eq:H2-anal-1}) and (\ref{eq:H2-anal-2}) hold.
\item
The inequalities
\begin{align}
& \hspace{-9mm}
P - \widetilde{A}^T (P\otimes I_{\bar{n}}) \widetilde{A} - E[C(\xi_0)^T C(\xi_0)]> 0, \label{eq:H2-anal-1-trac}\\
& \hspace{-9mm}
{\rm tr}\left(E[D(\xi_0)^T D(\xi_0) ]\right)+ {\rm tr}\left(E[B(\xi_0)B(\xi_0)^T] P\right) < \gamma^2 \label{eq:H2-anal-2-trac}
\end{align}
hold, where 
$\widetilde{A}$ is the matrix given by
\begin{align}
&\hspace{-5mm}
\widetilde{A} :=[\bar{A}_{1}^T, \ldots,
 \bar{A}_{n}^T]^T \in {\bf R}^{n\bar{n} \times n},\label{eq:def-tilX}\\
&\hspace{-5mm}
\bar{A}=:
\left[\bar{A}_{1}, \ldots, \bar{A}_{n}\right] \ \ 
(\bar{A}_{i} \in {\bf R}^{\bar{n}\times n}\ (i=1,\ldots,n))
\end{align}
with a matrix $\bar{A} \in {\bf R}^{\bar{n}\times n^2}\ (\bar{n}\leq n^2)$ satisfying
\begin{align}
&
\bar{A}^T \bar{A}=E\big[{\rm row}(A(\xi_0))^T
{\rm row}(A(\xi_0))\big].
\label{eq:equiv-rep-decom}
\end{align}
\end{enumerate}
\end{theorem}

\medskip

\begin{proof}
The equivalence between (\ref{eq:H2-anal-1}) and (\ref{eq:H2-anal-1-trac}) can be confirmed by using Lemma~2 in \cite{Hosoe-Auto20}.
The equivalence between (\ref{eq:H2-anal-2}) and (\ref{eq:H2-anal-2-trac}) can be confirmed by using the properties of the trace operation and its commutativity with the expectation operation.
\end{proof}

Since $\widetilde{A}$ is independent of the decision variable $P$, the inequality condition (\ref{eq:H2-anal-1-trac}) and (\ref{eq:H2-anal-2-trac}) can be viewed as a standard LMI, once all the expectations are computed.
Since the LMI (\ref{eq:H2-anal-1-trac}) and (\ref{eq:H2-anal-2-trac}) (with $P$ viewed as a decision variable) gives a necessary and sufficient condition for $H_2$ analysis, minimizing $\gamma$ with respect to the LMI leads us to the $H_2$ norm of $G$ without conservativeness.

\section{$H_2$ State-Feedback Controller Synthesis}
\label{sc:syn}

This section extends the results about $H_2$ analysis in the preceding section toward state feedback synthesis.

\subsection{Synthesis Problem}

We first describe our synthesis problem.
Let us consider the stochastic process $\xi$ satisfying Assumption~\ref{as:iid} and the associated generalized plant
\begin{align}
&
x_{k+1} = A_{\rm o}(\xi_k) x_k + B_{{\rm o}w}(\xi_k) w_k + B_{{\rm o}u}(\xi_k) u_k, \notag\\
&
z_{k} = C_{\rm o}(\xi_k) x_k + D_{{\rm o}w}(\xi_k) w_k + D_{{\rm o}u}(\xi_k) u_k, \label{eq:open-sys}
\end{align}
where
$A_{\rm o}: {\boldsymbol {\mathit\Xi}} \rightarrow {\bf R}^{n\times n}$, 
$B_{{\rm o}w}: {\boldsymbol {\mathit\Xi}} \rightarrow {\bf R}^{n\times p_w}$, 
$B_{{\rm o}u}: {\boldsymbol {\mathit\Xi}} \rightarrow {\bf R}^{n\times p_u}$, 
$C_{\rm o}: {\boldsymbol {\mathit\Xi}} \rightarrow {\bf R}^{q_z\times n}$, 
$D_{{\rm o}w}: {\boldsymbol {\mathit\Xi}} \rightarrow {\bf R}^{q_z\times p_w}$ and 
$D_{{\rm o}u}: {\boldsymbol {\mathit\Xi}} \rightarrow {\bf R}^{q_z\times p_u}$ are matrix-valued Borel functions,
and the initial state $x_0$ is supposed to be deterministic.
In this plant, $u_k$ is the control input.
As is the case with $G$ without the control input, we introduce the following assumption on the coefficient matrices of the plant.
\begin{assumption}
\label{as:inf-syn}
The squares of entries of
$A_{\rm o}(\xi_0)$, $B_{{\rm o}w}(\xi_0)$, $B_{{\rm o}u}(\xi_0)$, $C_{\rm o}(\xi_0)$, $D_{{\rm o}w}(\xi_0)$ and $D_{{\rm o}u}(\xi_0)$ are all Lebesgue integrable.
\end{assumption}

Let us next consider the state-feedback controller 
\begin{align}
&
u_k=F x_k \label{eq:state-feedback}
\end{align}
with the static time-invariant gain $F\in {\bf R}^{p_u\times n}$.
Then, the closed-loop system $G_F$ consisting of the plant 
(\ref{eq:open-sys}) and this controller
is described by (\ref{eq:fr-sys}) with
\begin{align}
&
A(\cdot)=A_{\rm o}(\cdot)+B_{{\rm o}u}(\cdot)F,\ \ B(\cdot)=B_{{\rm
 o}w}(\cdot),\notag \\
&
C(\cdot)=C_{\rm o}(\cdot)+D_{{\rm o}u}(\cdot)F,\ \ D(\cdot)=D_{{\rm
 o}w}(\cdot).
\label{eq:closed-loop}
\end{align}
Obviously, if the plant satisfies Assumption~\ref{as:inf-syn}, then the corresponding $G_F$ satisfies Assumption~\ref{as:bound} for each fixed gain $F$.
This section tackles the problem of designing a gain $F$ minimizing the $H_2$ norm $\|G_F\|_2$ of the corresponding closed-loop system $G_F$.

\subsection{Synthesis-Oriented Inequality Condition}

For given $F$, the $H_2$ norm of the corresponding closed-loop system $G_F$ can be analyzed as an LMI optimization problem by Theorems~\ref{th:main-H2anal} and \ref{th:anal-trac}.
In the synthesis, however, the gain $F$ is also viewed as a decision variable, and hence, the inequality condition (\ref{eq:H2-anal-1-trac}) and (\ref{eq:H2-anal-2-trac}) becomes nonlinear in the decision variables.
To make matters worse, the direct linearization of the inequality condition is not straightforward since the decision variables are involved in the matrix $\widetilde{A}$ in a complicated form.
Hence, we return to the expectation-based inequality condition (\ref{eq:H2-anal-1}) and (\ref{eq:H2-anal-2}), and derive a numerically tractable synthesis-oriented inequality condition from it.

The following theorem is a main result about the LMI condition for $H_2$ performance synthesis.
\begin{theorem}
\label{th:syn-trac}
Suppose that the generalized plant (\ref{eq:open-sys}) satisfies Assumptions~\ref{as:iid} and \ref{as:inf-syn}.
For given $\gamma>0$,
there exists a gain $F$ such that the closed-loop system $G_F$ is exponentially stable in the second moment and satisfies $\|G_F\|_2<\gamma$ if and only if there exist $X\in {\bf S}^{n\times n}_+$, $Y\in {\bf R}^{p_u \times n}$ and $R\in {\bf S}^{p_w \times p_w}$ satisfying
\begin{align}
&
\begin{bmatrix}
X & * & * \\
\widetilde{A}_{\rm o} X + \widetilde{B}_{{\rm o}u} Y & X\otimes I_{\bar{n}_u} & * \\
\widetilde{C}_{\rm o} X + \widetilde{D}_{{\rm o}u} Y & 0 & I_{q_z\bar{q}_{zu}}
\end{bmatrix}>0,  \label{eq:syn-cond1}\\
&
\begin{bmatrix}
R-E_D & * \\
\widetilde{B}_{{\rm o}w} & X\otimes I_{\bar{p}_w}
\end{bmatrix}>0,   \label{eq:syn-cond2}\\
&
{\rm tr}(R)<\gamma^2  \label{eq:syn-cond3}
\end{align}
($*$ denotes the transpose of an appropriate submatrix), where
\begin{align}
E_D:=E[D_{{\rm o}w}(\xi_0)^T D_{{\rm o}w}(\xi_0)],
\label{eq:EDD}
\end{align}
and
$\widetilde{A}_{\rm o}$, $\widetilde{B}_{{\rm o}u}$,
$\widetilde{C}_{\rm o}$, $\widetilde{D}_{{\rm o}u}$ and
$\widetilde{B}_{{\rm o}w}$ are the matrices given by
\begin{align}
&
\widetilde{A}_{{\rm o}} :=[\bar{A}_{{\rm o}1}^T, \ldots,
 \bar{A}_{{\rm o}n}^T]^T \in {\bf R}^{n\bar{n}_u \times n},\notag \\
 &
\widetilde{B}_{{\rm o}u} :=[\bar{B}_{{\rm o}u1}^T, \ldots,
 \bar{B}_{{\rm o}un}^T]^T \in {\bf R}^{n\bar{n}_u \times p_u},\notag \\
&
\bar{G}_{AB}=:
\left[\bar{A}_{{\rm o}1}, \ldots, \bar{A}_{{\rm o}n},
\bar{B}_{{\rm o}u1}, \ldots, \bar{B}_{{\rm o}un}\right] \notag \\
&
(\bar{A}_{{\rm o}i} \in {\bf R}^{\bar{n}_u\times n}, 
\bar{B}_{{\rm o}ui} \in {\bf R}^{\bar{n}_u\times p_u}\ (i=1,\ldots,n)),\label{eq:tilAB}\\
%
%
&
\widetilde{C}_{{\rm o}} :=[\bar{C}_{{\rm o}1}^T, \ldots,
 \bar{C}_{{\rm o}q_z}^T]^T \in {\bf R}^{q_z\bar{q}_{zu} \times n},\notag\\
 &
\widetilde{D}_{{\rm o}u} :=[\bar{D}_{{\rm o}u1}^T, \ldots,
 \bar{D}_{{\rm o}u q_z}^T]^T \in {\bf R}^{q_z \bar{q}_{zu} \times p_u},\notag\\
&
\bar{G}_{CD}=:
\left[\bar{C}_{{\rm o}1}, \ldots, \bar{C}_{{\rm o}q_z},
\bar{D}_{{\rm o}u1}, \ldots, \bar{D}_{{\rm o}u q_z}\right] \notag \\
&
(\bar{C}_{{\rm o}i} \in {\bf R}^{\bar{q}_{zu}\times n}, 
\bar{D}_{{\rm o}ui} \in {\bf R}^{\bar{q}_{zu}\times p_u}\ (i=1,\ldots,q_z)),\label{eq:tilCD}\\
%
%
&
\widetilde{B}_{{\rm o}w} :=[\bar{B}_{{\rm o}w1}^T, \ldots,
 \bar{B}_{{\rm o}wn}^T]^T \in {\bf R}^{n\bar{p}_w \times p_w},\notag\\
&
\bar{B}_{{\rm o}w}=:
\left[\bar{B}_{{\rm o}w1}, \ldots, \bar{B}_{{\rm o}wn}\right] \notag \\
& 
(\bar{B}_{{\rm o}wi} \in {\bf R}^{\bar{p}_w\times p_w}\ (i=1,\ldots,n))\label{eq:tilB}
\end{align}
with matrices 
$\bar{G}_{AB} \in {\bf R}^{\bar{n}_u\times (n+p_u)n}\ (\bar{n}_u\leq (n+p_u)n)$,
$\bar{G}_{CD} \in {\bf R}^{\bar{q}_{zu}\times (n+p_u)q_z}\ (\bar{q}_{zu}\leq (n+p_u)q_z)$
and
$\bar{B}_{{\rm o}w} \in {\bf R}^{\bar{p}_w\times np_w}\ (\bar{p}_w\leq np_w)$ satisfying
\begin{align}
&
\bar{G}_{AB}^T \bar{G}_{AB}=
E\big[[{\rm row}(A_{{\rm o}}(\xi_0)), {\rm row}(B_{{\rm o}u}(\xi_0))]^T\notag\\
&\hspace{25mm}\cdot
[{\rm row}(A_{{\rm o}}(\xi_0)), {\rm row}(B_{{\rm o}u}(\xi_0))]\big],
\label{eq:EBAAB}\\
%
%
%
&
\bar{G}_{CD}^T \bar{G}_{CD}=
E\big[[{\rm row}(C_{{\rm o}}(\xi_0)), {\rm row}(D_{{\rm o}u}(\xi_0))]^T\notag\\
&\hspace{25mm}\cdot
[{\rm row}(C_{{\rm o}}(\xi_0)), {\rm row}(D_{{\rm o}u}(\xi_0))]\big],
\label{eq:EDCCD}\\
%
%
&
\bar{B}_{{\rm o}w}^T \bar{B}_{{\rm o}w}=E\big[{\rm row}(B_{{\rm o}w}(\xi_0))^T
{\rm row}(B_{{\rm o}w}(\xi_0))\big].
\label{eq:EBB}
\end{align}
In particular, $F=YX^{-1}$ is one such gain.
\end{theorem}

\begin{proof}
By theorem~\ref{th:main-H2anal}, the inequality condition for $H_2$ performance of $G_F$ is given by
\begin{align}
&
E[P - (A_{\rm o}(\xi_0)+B_{{\rm o}u}(\xi_0)F)^T P (A_{\rm o}(\xi_0)+B_{{\rm o}u}(\xi_0)F) \notag \\
&
- (C_{\rm o}(\xi_0)+D_{{\rm o}u}(\xi_0)F)^T (C_{\rm o}(\xi_0)+D_{{\rm o}u}(\xi_0)F)]> 0, \label{eq:syn-pf1}\\
& 
E[{\rm tr}(D_{{\rm o}w}(\xi_0)^T D_{{\rm o}w}(\xi_0) + B_{{\rm o}w}(\xi_0)^T P B_{{\rm o}w}(\xi_0))] < \gamma^2.
\label{eq:syn-pf2}
\end{align}
Hence, it suffices for the present proof that for given $\gamma>0$, the existence of $P\in {\bf S}^{n\times n}_+$ and $F$ satisfying these two inequalities is equivalent to that of $X\in {\bf S}^{n\times n}_+$, $Y\in {\bf R}^{p_u \times n}$ and $R\in {\bf S}^{p_u \times p_u}$
satisfying (\ref{eq:syn-cond1})--(\ref{eq:syn-cond3}).

As is the case with aforementioned analysis, the expectation operation in (\ref{eq:syn-pf1}) and (\ref{eq:syn-pf2}) can be pre-calculated.
That is, by using Lemma~2 in \cite{Hosoe-Auto20}, those inequalities can be equivalently rewritten as
\begin{align}
&
P - (\widetilde{A}_{\rm o}+\widetilde{B}_{{\rm o}u}F)^T (P\otimes I_{\bar{n}_u}) (\widetilde{A}_{\rm o}+\widetilde{B}_{{\rm o}u}F) \notag \\
&
- (\widetilde{C}_{\rm o}+\widetilde{D}_{{\rm o}u}F)^T (\widetilde{C}_{\rm o}+\widetilde{D}_{{\rm o}u}F)> 0, \label{eq:syn-pf3}\\
& 
{\rm tr}(E_D + \widetilde{B}_{{\rm o}w}^T (P\otimes I_{\bar{p}_w}) \widetilde{B}_{{\rm o}w}) < \gamma^2,
\label{eq:syn-pf4}
\end{align}
where the coefficient matrices are given by (\ref{eq:EDD})--(\ref{eq:EBB}).
For given $P$, the satisfaction of (\ref{eq:syn-pf4}) is equivalent to the existence of $R\in {\bf S}^{p_u \times p_u}$ satisfying 
(\ref{eq:syn-cond3}) and
\begin{align}
&
E_D + \widetilde{B}_{{\rm o}w}^T (P\otimes I_{\bar{p}_w}) \widetilde{B}_{{\rm o}w} <R.
\label{eq:syn-pf5}
\end{align}
Then, it follows from the Schur complement technique that 
(\ref{eq:syn-pf3}) and (\ref{eq:syn-pf5}) are respectively equivalent to
\begin{align}
&
\begin{bmatrix}
P & * & * \\
(P\otimes I_{\bar{n}_u})(\widetilde{A}_{\rm o}  + \widetilde{B}_{{\rm o}u} F) & P\otimes I_{\bar{n}_u} & * \\
\widetilde{C}_{\rm o}  + \widetilde{D}_{{\rm o}u} F & 0 & I_{q_z\bar{q}_{zu}}
\end{bmatrix}>0,  \\
&
\begin{bmatrix}
R-E_D & * \\
(P\otimes I_{\bar{p}_w})\widetilde{B}_{{\rm o}w} & P\otimes I_{\bar{p}_w}
\end{bmatrix}>0
\end{align}
for given $P\ (>0)$, $F$ and $R$.
Hence,  congruence transformation with appropriate matrices and the change of variables $X=P^{-1}$ and $Y=FX$ lead us to 
(\ref{eq:syn-cond1}) and (\ref{eq:syn-cond2}).
This completes the proof.
\end{proof}

If the coefficient matrices of the plant (\ref{eq:open-sys}) are all deterministic, then $\bar{n}_u$, $\bar{q}_{zu}$ and $\bar{p}_w$ become $1$, and the corresponding inequality condition (\ref{eq:syn-cond1})--(\ref{eq:syn-cond3}) reduces to that for deterministic systems.
Hence, the present result is a solid extension of the conventional result for deterministic systems.
A similar comment also applies to the results about $H_2$ analysis discussed in the preceding section.

\section{Numerical Example}
\label{sc:exam}

\begin{figure}[b] 
   \centering
   \vspace{-5mm}
      \includegraphics[width=3.5in]{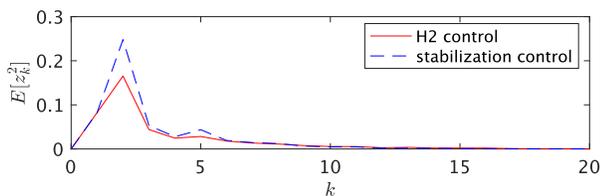}%
         \vspace{-4mm}
   \caption{Impulse responses with two gains.}%
   \label{fig:impulse}%
\end{figure}%

Let us consider the system (\ref{eq:open-sys}) with 
\begin{align}
&
A_{\rm o}(\xi_k)=
\begin{bmatrix}
1.3+\xi_{2k}&  0.8+\xi_{1k}    &  -0.5\\
   0.5  &    0.3+\xi_{1k}\xi_{2k} & -1.2+\xi_{1k}^2\\
   -0.2  &   0.8      &    0.6
\end{bmatrix},\notag\\
&
B_{{\rm o}w}(\xi_k)=B_{{\rm o}u}(\xi_k)=\begin{bmatrix}
0 & 0 & 1
  \end{bmatrix}^T,\notag \\
  & 
C_{{\rm o}}=\begin{bmatrix}
0 & \xi_{1k} & \xi_{2k}
  \end{bmatrix},\ \ 
 D_{{\rm o}w}
 =D_{{\rm o}u}
 =0, \notag \\
 &
\xi_{1k}\sim {\cal N}(0,0.2^2),\ \ \xi_{2k}\sim {\cal U}(-0.5,0.5),
\end{align}
where ${\cal N}(0,0.2^2)$ and ${\cal U}(-0.5,0.5)$ are the normal distribution with mean $0$
and standard deviation $0.2$ and the continuous uniform distribution
with minimum $-0.5$ and maximum $0.5$, respectively.
This system is unstable.

For the above system, we designed two types of controllers: one is by the $H_{2}$ control in this paper, and the other is by the stabilization control in \cite{Hosoe-TAC19}.
Through minimizing $\gamma$ with respect to the LMI in Theorem~\ref{th:syn-trac}, we obtained $F=[1.6739,    0.1027,   -1.7100]$ with the minimal value of $\gamma=0.6520$, which corresponds to the closed-loop $H_{2}$ performance $\|G_{F}\|_{2}$.
On the other hand, through minimizing $\lambda$ in (\ref{eq:exp-def}) by the approach in \cite{Hosoe-TAC19}, we obtained $F=[2.1622,    0.4018,   -2.0782]$ with the minimal value of $\lambda=0.8385$.
To confirm the effectiveness of the $H_{2}$ control, we compare the impulse responses of the closed-loop systems with these two gains.
We generated $10^{3}$ sample paths of $\xi$, and calculated the corresponding output $z$ with each gain under the impulse input $w$.
Then, we obtained the time sequences of $E[z_{k}^{2}]$ shown in Fig.~\ref{fig:impulse}, where the expectation was approximated by the sample mean using the $10^{3}$ sample paths.
From this figure, we see that the transient response was successfully improved (i.e., became smaller in the sense of $H_{2}$ norm $\|G_{F}\|_{2}=\sum_{k=0}^{\infty}E[z_{k}^{2}]$) by the $H_{2}$ control.

%
%
%
%
%

\section{Conclusions}
\label{sc:concl}

This paper derived numerically tractable LMI conditions for $H_{2}$ performance analysis and controller synthesis for discrete-time linear systems with dynamics determined by an i.i.d.\ process.
Deriving LMI conditions for other control problems about the systems as well as
applying the results to remote control of vehicles are possible future works.



\bibliographystyle{IEEEtran}
\bibliography{rds}

%
%
%
%

\end{document}